\documentclass[runningheads,a4paper]{llncs}
\usepackage{amsmath}
\usepackage{amssymb}
\usepackage{graphicx}
\usepackage{cite}
\usepackage{todonotes}
\usepackage{subfig}
\usepackage{float}

\graphicspath{{pics/}}

\usepackage{url}
\urldef{\mailsa}\path|{loboda, alserg}@rain.ifmo.ru|
\urldef{\mailsb}\path|martyomov@pathology.wustl.edu|
\newcommand{\keywords}[1]{\par\addvspace\baselineskip
\noindent\keywordname\enspace\ignorespaces#1}
\begin{document}
\mainmatter  

\title{Solving generalized maximum-weight connected 
subgraph problem for network enrichment analysis}

\titlerunning{
    Solving generalized maximum-weight connected 
subgraph problem}

%
\author{Alexander A. Loboda\inst{1}%
\and Maxim N. Artyomov\inst{2} \and Alexey A. Sergushichev\inst{1}}
\authorrunning{
    Solving generalized maximum-weight connected 
subgraph problem}

\institute{Computer Technologies Department, ITMO University, 
    Saint Petersburg, 197101, Russia \\
\mailsa \and 
Department of Pathology and Immunology, Washington University in St. Louis, 
    MO, USA\\
\mailsb
}

%
%

\toctitle{Lecture Notes in Computer Science}
\tocauthor{Authors' Instructions}
\maketitle

\begin{abstract}
Network enrichment analysis methods allow to identify active modules
without being biased towards \emph{a priori} defined pathways.
One of mathematical formulations of such analysis is a reduction 
to a maximum-weight connected subgraph problem. 
In particular, in analysis of metabolic networks 
a generalized maximum-weight
connected subgraph (GMWCS) problem, where both nodes and edges are scored,
naturally arises. Here we present the first to our knowledge practical exact
GMWCS solver. We have tested it on real-world instances
and compared to similar solvers.
First, the results show that on node-weighted instances 
GMWCS solver has a similar performance to the best solver 
for that problem. Second, GMWCS solver is faster 
compared to the closest analogue when run on GMWCS instances with 
edge weights.
\end{abstract}

\section{Introduction}

Gene set enrichment methods are widely used for the analysis of untargeted
biological data such as transcriptomic, proteomic or metabolomic profiles.
These methods allow to identify molecular pathways, in a form of gene sets,
that have non-random group behaviour in the data. Determining such
overenriched pathways provides insights into the data and allows to
better understand the considered system.

Network enrichment methods, in opposite to gene set enrichment, do not rely on
the predefined gene sets and, thus, allow to identify novel pathways. These
methods use network of interacting entities, such as genes, proteins,
metabolites, etc. and try to identify the most regulated subnetwork.
There are different mathematical formulations of the network enrichment
problem, but many of them are NP-hard \cite{Ideker2002, Dittrich2008a, Alcaraz2014}. 

Dittrich et al. in \cite{Dittrich2008a} suggested a formulation as a
maximum-weight
connected subgraph (MWCS) problem. Originally, the authors considered
node-weighted graph, such that positive weight corresponded to "interesting"
nodes and negative weight corresponded to "non-interesting" nodes.  The goal
was to find a connected graph with the maximal sum of weights of its nodes,
which corresponded to an "active module".

Here we consider, a slightly different form of MWCS, generalized MWCS (GMWCS),
that has edges also weighted. Such formulation naturally arises in the studies of
metabolic networks \cite{Beisser2012,Sergushichev2016}, where nodes in the graph
represent metabolites and edges represent their interconversions via reactions.
There, the nodes can be scored using metabolomic profiles and the edges can be
scored using gene or protein expression profiles.

In this paper we describe an exact solver for the node-and-edge-weighted GMWCS
problem. First, in section~\ref{sec:defs} we give formal definitions. 
Then in section~\ref{sec:preprocessing} we describe preprocessing steps adapted for the
edge-based formulation. In section~\ref{sec:decomposition} we show how the
instance can be split into three smaller instances. Section~\ref{sec:mip} is
dedicated to a mixed-integer programming (MIP) formulation of the problem.  In
section~\ref{sec:experiments} we show experimental results of running
the solver on real-world instances that
appear in GAM web-service and show that it is faster and more accurate
than \emph{Heinz}~\cite{Beisser2012a} on edge-weighted instances and is similar in performance to
\emph{Heinz2}~\cite{El-Kebir2014} on node-weighted instances.

\section{Formal definitions}
\label{sec:defs}

Here we consider the Maximum-Weight Connected Subgraph (MWCS) problem for
which there are two slightly different formulations. In the most commonly used
definition of MWCS only nodes are weighted \cite{Alvarez2013, El-Kebir2014}.
In this paper we consider problem where edges are weighted too
\cite{Haouari2013}. To remove the ambiguity we call 
the former problem Simple MWCS (SMWCS) and the latter one
Generalized MWCS (GMWCS).

The goal of MWCS problems is to find in a given graph a connected subgraph with the maximal 
the maximal sum of weights. As a subgraph is connected we can consider
connected components of the graph independently. Thus, below we assume that 
the input graph is connected. 

First, we give definition of a Simple Maximum-Weight Connected Subgraph
problem. 

\begin{definition}
    Given a connected undirected graph $G = (V, E)$ and weight function
    $\omega_v: V \rightarrow \mathbb{R}$, the Simple Maximum-Weight Connected Subgraph
    (SMWCS) problem is the problem of finding a connected subgraph
    $\widetilde{G} = (\widetilde{V}, \widetilde{E})$ with the maximal total
    weight 

    $$ 
    \Omega(\widetilde{G}) = \sum_{v \in \widetilde{V}} \omega(v) 
    \rightarrow max
    $$ 
    \end{definition}

Second, we define generalized variant of this problem, where both nodes
and edges could be weighted.

\begin{definition}
    Given a connected undirected graph $G = (V, E)$ and a weight function
    $\omega: (V \cup E) \rightarrow \mathbb{R}$, the Generalized Maximum-Weight
    Connected Subgraph (GMWCS) problem is the problem of finding a connected
    subgraph $\widetilde{G} = (\widetilde{V},\widetilde{E})$ with the maximal
    total weight

    $$ 
    \Omega(\widetilde{G}) = \sum_{v \in \widetilde{V}} \omega(v) + 
    \sum_{e \in \widetilde{E}} \omega(e) \rightarrow max
    $$ 

\end{definition}

Now we define a rooted variant of the problem with one of the vertices forced
to in a solution. It is used as an auxiliary subproblem of GMWCS.

\begin{definition}
    Given a connected undirected graph $G = (V, E)$, a weight function $\omega:
    (V \cup E) \rightarrow \mathbb{R}$ and a root node $r \in V$ the Rooted
    Generalized Maximum-Weight Connected Subgraph (R-GMWCS) problem is the problem of
    finding a connected subgraph $\widetilde{G} = (\widetilde{V},\widetilde{E})$ such that
    $r \in \widetilde{V}$ and

    $$ 
    \Omega(\widetilde{G}) = \sum_{v \in \widetilde{V}} \omega(v) + 
    \sum_{e \in \widetilde{E}} \omega(e) \rightarrow max
    $$ 
\end{definition}

El-Kebir et al. in \cite{El-Kebir2014} have shown that MWCS problem is NP-hard.
Since MWCS is a special case of GMWCS then GMWCS is also NP-hard. R-GMWCS 
problem is NP-hard too because any instance of GMWCS problem can be solved
by solving an R-GMWCS instance for each node as a root.

Finally, below we use $n$ as a shorthand for the number of nodes $|V|$ in
the graph $G$.

\section{Preprocessing}
\label{sec:preprocessing}

We introduce two preprocessing rules adapted from \cite{El-Kebir2014} that
simplify the problem. These rules make a new graph with a smaller number of
vertices and edges in such a way that the GMWCS solution for the original graph can be
easily recovered from the GMWCS solution for the simplified graph.

First, we merge groups of close vertices that either none or all of them are in
the optimal solution (Fig.~\ref{fig:preprocessing}A).
Let $e = (u, v)$ be an edge with $\omega(e) \ge 0$ with
simultaneously $\omega(e) + \omega(v) \ge 0$ and $\omega(e) + \omega(u) \ge 0$.
In this case if one of the vertices is included in the solution then the
edge and the
other vertex can also be included without decreasing the total weight.
Thus, we can contract edge $e$ into a new vertex $w$ with
a weight $\omega(w) = \omega(e) + \omega(u) + \omega(v).$ After the contraction
parallel edges between $w$ and some vertex $t$ could appear. In that case we
merge all non-negative one into a single edge with weight of the sum of their
weights. After that, we remove all edges between $w$ and $t$ except one with
the maximal weight. We try to apply this rule for all vertices in the loop
while the graph is changing.

\begin{figure}[H]
    \centering
    {\includegraphics[width=12cm]{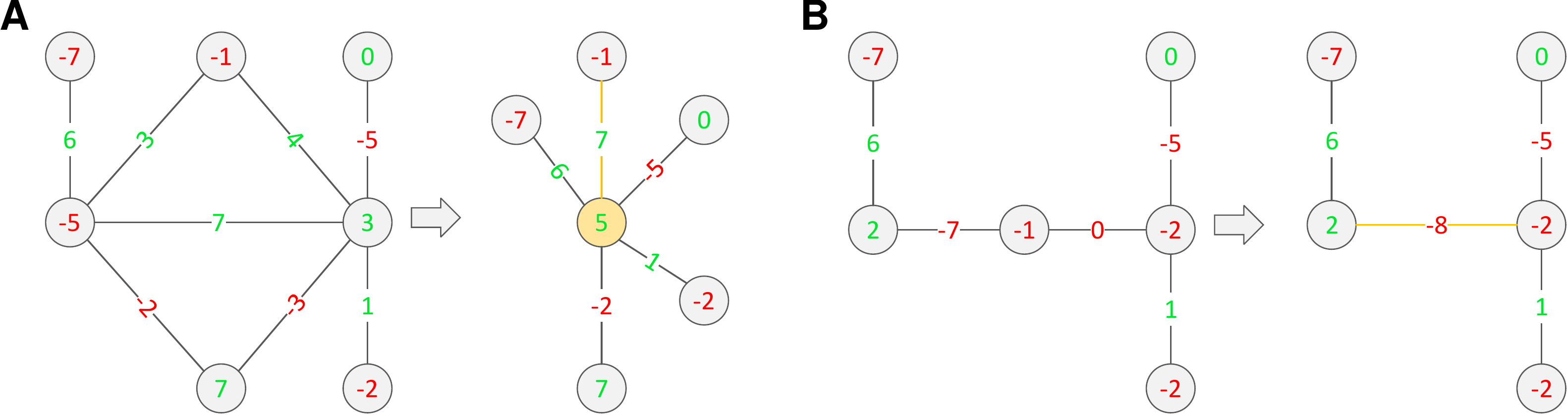} }
    \caption{Applying first rule that contract an edge (A) and second rule 
    that replace negative chain by a single edge (B). New vertices and nodes 
    painted yellow}%
    \label{fig:preprocessing}%
\end{figure}

Second, similarly to the previous step, we merge negative-weighted chains
(Fig.~\ref{fig:preprocessing}B).  Let
$v$ be a vertex with $deg(v) = 2$ with corresponding incident edges 
$e_1 = (u, v)$ and $e_2 = (v, w)$. 
If all three weights $\omega(v)$, $\omega(e_1)$
and $\omega(e_2)$ are negative, then $v$, $e_1$ and $e_2$ could be replaced
with a single edge $e = (u, w)$ with a weight $\omega(e) = \omega(v)
+ \omega(e_1) + \omega(e_2)$.  Merging negative chains is implemented in
a single pass by iteratively trying to apply the rule for all the nodes.

\section{Cut vertex decomposition}
\label{sec:decomposition}

In this section we discuss how a GMWCS instance can be decompose into three
smaller problems. The decomposition is based on the idea that 
biconnected components can be considered separately \cite{El-Kebir2014}.

\begin{figure}[H]
    \centering
    \includegraphics[width=7cm]{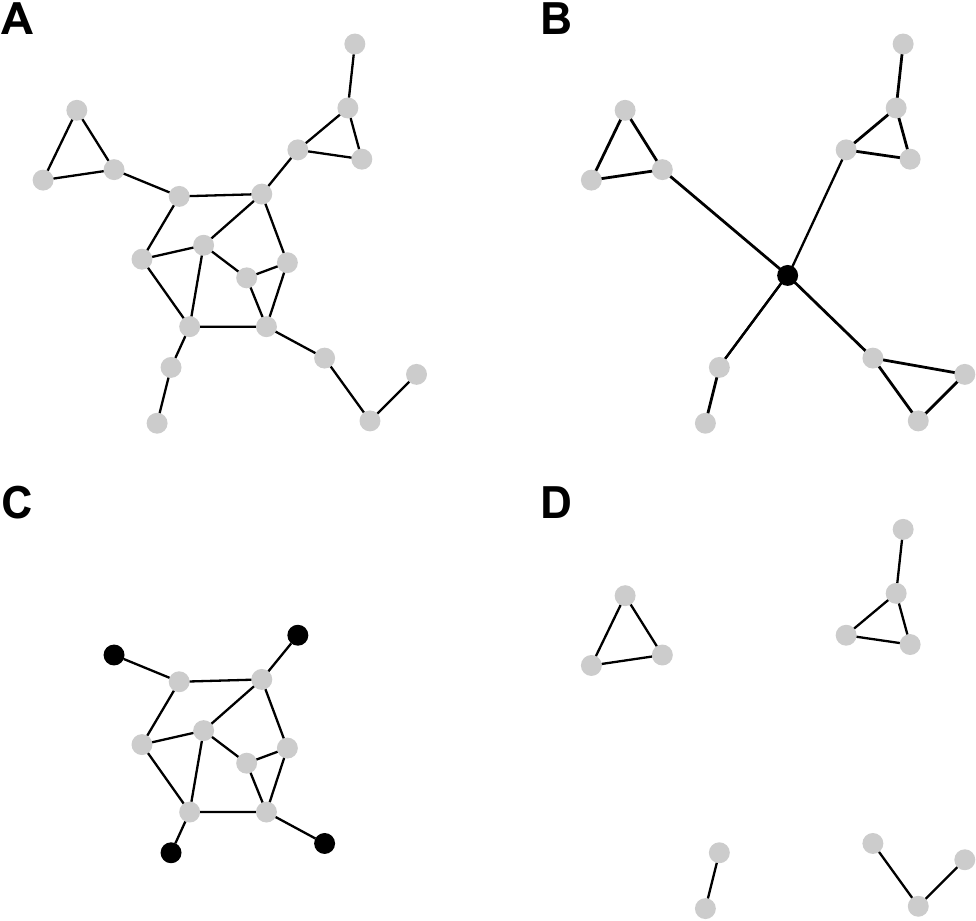}%
    \caption{Input graph and instances spawned by decomposition}
    \label{fig:decomposition}
\end{figure}

Briefly, we have  a GMWCS instance as input (Fig.~\ref{fig:decomposition}A).
First, we merge the largest biconnected component 
into a single vertex with zero weight
and solve an R-GMWCS instance for this modified graph and the new vertex 
as a root~(Fig.~\ref{fig:decomposition}B).
Then, we replace each of the components branching from the largest biconnected component
by a single vertex with weight equal to the weight of corresponding
subgraph in the R-GMWCS solution from the previous step~(Fig.~\ref{fig:decomposition}C).
Last, we try to find a subgraph with a greater weight which fully lies
in one of the branching components~(Fig.~\ref{fig:decomposition}D).

Formally, let $B$ be a biconnected component of the graph $G$ with the maximal
number of vertices. Let $C$ be a set of cut vertices of the graph $G$ that are
also contained in $B$. Let $B_c$ be a component containing $c$ in the graph
$G \setminus (B \setminus C)$.

\begin{proposition}
    Let a subgraph $\widetilde{G}$ of $G$ be an optimal solution of GMWCS for
    graph $G$ and $\widetilde{G}_c$, $\forall c \in C$, are optimal solutions for R-GMWCS instances
    for graphs $B_c$ with a root $c$.  In this case, if $\widetilde{G}$
    contains a vertex $c \in C$, then we can construct an optimal
    solution $\widetilde{G}'$ such that:
1) $\widetilde{G}' \cap B$ = $\widetilde{G} \cap B$ and 
2) $\widetilde{G}' \cap B_c = \widetilde{G}_c$.
\label{prop:cutvertex}
\end{proposition}
\begin{proof}
    Let $\widetilde{B}_c = \widetilde{G} \cap B_c$. We prove that it can be
    replaced by $\widetilde{G}_c$ without loss of connectivity and optimality.
    First, $\widetilde{B}_c$ must be connected. Let it be disconnected. Then
    there is no path between $c$ and some vertex $v$. Since $\widetilde{G}$ is
    connected then there is a simple path $vc$ in $G$. However, by definition 
    of cut vertex, path $vc$ can not contain vertices from $G \setminus B_c$
    and, thus it fully lies in $B_c$, a contradiction. 
    Since $\widetilde{B}_c$ is connected and contains $c$ then it cannot
    have weight greater than $\widetilde{G}_c$ by construction of
    $\widetilde{G}_c$.

    Now we prove that the replacement keeps the graph connected.
    Repeating the reasoning from the previous step we can get that
    $\widetilde{G} \cap B$ must
    be connected. So,
    $\widetilde{G}_c$ is connected,
    $\widetilde{G} \cap B$ connected
    and both these graphs contain $c$. Thus, $\widetilde{G}'$ is also 
    connected.
    \hfill$\square$

\end{proof}

This proposition allows us to consider only optimal solutions that either
include a vertex from $B$ and in subgraphs $B_c$ are identical to the 
corresponding R-GMWCS instance or fully lie in some of the subgraphs $B_c$.

First, for each $c \in C$ we want to know the best solution of the problem for
the graph $B_c$ containing vertex $c$. It is precisely an R-GMWCS instance.
For practical reasons, it is better to spawn one instance at this step instead of $|C|$
instances.  Let $G^* = \bigcup_{\forall c \in C}B_c$. Then we merge
all vertices from $C$ contained in $G^*$ into a single vertex $r$ 
with $\omega(r) = 0$ and solve R-GMWCS problem for such graph. Let $S$ to be the solution of
this instance.  To get solution for the graph $B_c$ we replace
back $r$ to $c$ in $S$, and remove all the vertices which are not contained in
$B_c$.

Second, we find best scored subgraph of $G$ that do not lies fully in some of
$B_c$. Let $\widetilde{G}_{c}$ be the solution of R-GMWCS for graph $B_c$ with
root $c$ obtained on the previous step. We obtain a new GMWCS instance
by considering the component $B$ and for all $c \in C$ attaching 
a vertex $v$ with weight $\omega(v) = \Omega(\widetilde{G}_c)$. We solve the 
resulting instance and then recover a solution for  the original
problem.

Last, we find all potential solutions that fully lie in $B_c$ for all $c
\in C$.  For this purpose we spawn one instance for the graph $G^*
= \bigcup_{\forall c \in C} B_c$. Clearly that if the solution of the
problem for the graph $G$ lies fully in some of $B_c$ then we will find it at
this step.

\section{Mixed integer programming formulation}
\label{sec:mip}

Here we describe a MIP formulation of the problem. The GMWCS can be represented
as two parts: objective function (weight of the subgraph) that should be
maximized and constraints that ensure that the subgraph is connected. The
objective function is linear and can be put into a MIP problem in
a straightforward way. However, getting effective linear subgraph connectivity
constraints is not trivial. In this section we describe how it can be done.
The resulting MIP problem is solved by IBM ILOG CPLEX.

First, we consider a nonlinear formulation of the GMWCS problem, as proposed in
\cite{Haouari2013}. Then, we show how to eliminate nonlinearity and get
a linear system. Finally, we introduce extra symmetry-breaking and cuts,
which do not impact on the correctness of the formulation, but improve the
performance.

\subsection{Subgraph representation}
We use one binary variable for each vertex or edge that represent the presence
in the subgraph:
\begin{enumerate}

    \item Binary variable $y_v$ takes the value of 1 iff $v \in V$ belongs to the
subgraph.
    \item Binary variable $w_e$ takes the value of 1 iff $e \in E$ belongs to the
subgraph.
\end{enumerate}

For these variables to be representing a valid subgraph (not necessarily
connected) we need to introduce a set of constraints:
\begin{align}
    \label{in:edge}
    w_e &\le y_v, &\forall v \in V, e \in \delta_v.
\end{align}
These constraints state that an edge can be a part of the subgraph, 
only if both of its endpoints are a part of the subgraph.

\subsection{Nonlinear formulation}

The nonlinear formulation of the subgraph connectivity constraints is based on
the idea that any connected graph can be traversed from any of its vertices.
The output of the traversal can be represented as an arborescence where an arc
$(v,u)$ denotes that $v$ has been visited before $u$. Accordingly, we can
ensure connectivity
of a subgraph if we can provide an arborescence corresponding to the traversal of
this subgraph.

For a given graph $G = (V, E)$, let $S = (V,A)$ be a directed graph, where $A$
is obtained from $E$ by replacing each undirected edge $e=(v,u)$ 
by two directed arcs $(v,u)$ and $(u,v)$.

Now, we are going to introduce variables that we will use in the formulation
and show nonlinear system of constraints, that ensure connectivity of subgraph:
\begin{enumerate}
\item Binary variable $x_a$ takes the value of 1 iff $a \in A$ belongs to the 
arborescence.

\item Binary variable $r_v$ takes the value of 1 iff $v \in V$ is the root of
the arborescence.

\item Continuous variable $d_v$ takes the value of $n$ if the path in the
arborescence from the root to vertex $v$ contains $n$ vertices. If $v$ does not
belong to the solution then value can be arbitrary.
\end{enumerate}

Then we introduce constraints that ensure the validity of an arborescence:

\begin{align}
    \label{in:oneroot}
    \sum_{v \in V} r_v &= 1; \\
    \label{in:distlimits}
    1 \le d_v &\le n, &\forall v \in V; \\
    \label{in:onein}
    \sum_{(u, v) \in A} x_{uv} + r_v &= y_v, &\forall v \in V; \\
    \label{in:treeoversg}
    x_{vu} + x_{uv} &\le w_e, &\forall e = (v, u) \in E; \\
    \label{in:rootdist}
    d_vr_v &= r_v, &\forall v \in V; \\
    \label{in:dist}
    d_ux_{vu} &= (d_v + 1) x_{vu}, &\forall (v, u) \in A.
\end{align}

Inequality \eqref{in:oneroot} states that there is only one root in the
arborescence; \eqref{in:distlimits} is a limitation on the distance between any
vertex and the root; \eqref{in:onein} states that if a vertex is a part of the
subgraph then either it is a root of the arborescence or $deg_{in}(v)=1$;
\eqref{in:treeoversg} says that an arc of the arborescence can be in the 
solution only if the corresponding edge is also in it. 
Last two inequalities \eqref{in:rootdist} and \eqref{in:dist}
control correct distances in the arborescence.

Haouari et al. have shown in \cite{Haouari2013} that this nonlinear system is
a correct formulation of GMWCS. That is, the arborescence covers all vertices of
the resulting subgraph and the solution can induce this arborescence.

However, inequalities \eqref{in:rootdist} and \eqref{in:dist} are not linear
and should be replaced, so that the formulation can be represented as a MIP
problem. 

\subsection{Linearization}

Nonlinear equations \eqref{in:rootdist} and \eqref{in:dist} can be replaced
with the following system of linear inequalities:
\begin{align}
    \label{in:linrootdist}
    d_v + nr_v &\le n, &\forall v \in V; \\
    \label{in:lindist1}
    n + d_u - d_v &\ge (n + 1)x_{vu}, &\forall (v, u) \in A; \\
    \label{in:lindist2}
    n + d_v - d_u + &\ge (n - 1)x_{vu}, &\forall (v, u) \in A.
\end{align}

\begin{proposition}
    Every feasible solution to \eqref{in:edge}-\eqref{in:dist} is also feasible
    to \eqref{in:edge}-\eqref{in:treeoversg},
    \eqref{in:linrootdist}-\eqref{in:lindist2} and vice versa.
\end{proposition}

\begin{proof}
    First, we prove that \eqref{in:linrootdist} is equivalent to
    \eqref{in:rootdist} in a sense of feasibility of the solution. Since $r_v$
    is a binary variable, we can consider two cases. Suppose that $r_v = 1$,
    then \eqref{in:rootdist} will take the form $d_v = 1$ while
    \eqref{in:linrootdist} will take the from $d_v \le 1$, and with
    \eqref{in:distlimits} we have $d_v = 1$. Now suppose that $r_v = 0$,
    \eqref{in:rootdist} will look $0 = 0$, it means that in this case there is
    no additional restrictions on variables and \eqref{in:linrootdist} will take
    the form $d_v \le n$, but system already have such inequality. Thus
    \eqref{in:rootdist} and \eqref{in:linrootdist} are equivalent for 
    both possible values of $r_v$.

    At the second part of the proof we will use the same approach. Here we
    prove that \eqref{in:dist} can be represented as linear inequalities
    \eqref{in:lindist1} and \eqref{in:lindist2}. 
    \begin{enumerate}
        \item Let $x_{vu} = 1$. Then after substitution into \eqref{in:dist} 
            we have  $d_u = d_v + 1$. Then we substitute $x_{vu}$ into 
            \eqref{in:lindist1} and \eqref{in:lindist2}
            \begin{align*}
                n + d_u - d_v \ge n + 1 \\
                n + d_v - d_u \ge n - 1
            \end{align*}
            or, equivalently,
            \begin{align*}
                d_u \ge d_v + 1 \\
                d_v + 1 \ge d_u
            \end{align*}
            or $d_u = d_v + 1$.

        \item Let $x_{vu} = 0$. The original nonlinear equation will take the form
            $0 = 0$. As mentioned above, it means that there is no additional
            restrictions on variables. We have to show that \eqref{in:lindist1} and
            \eqref{in:lindist2} also do not add such restrictions.
            After substitution these inequalities take the form:
            \begin{align*}
                n + d_u - d_v \ge 0 \\
                n + d_v - d_u \ge 0
            \end{align*}
            or $|d_v - d_u| \le n$. Obviously, variables that
            hold \eqref{in:distlimits} automatically hold such inequality. Thus,
    additional restrictions have not be added. \hfill $\square$ 
    \end{enumerate}

\end{proof}

\subsection{Symmetry-breaking}

It is a common practice to decrease the number of feasible solutions by
limiting the number of different but logically equivalent feasible solutions.
Such solutions are called symmetric. 
In our formulation constraints
\eqref{in:edge}-\eqref{in:treeoversg},
    \eqref{in:linrootdist}-\eqref{in:lindist2}
allow any arborescence of the graph to show
its connectivity.  So, in this section we show how to decrease the number of
feasible arborescences and thus decrease the search space.

\subsubsection{Root order rule.}
\label{subsec:rootsb}

First of all, for the unrooted GMWCS problem we force the arborescence root to
be a vertex with the maximal weight among present in the subgraph. Corresponding
constraint that is added in the MIP instance is:
\begin{align}
    \label{in:maxroot}
    \sum_{v \prec u} r_v &\le 1 - y_u, &\forall u \in V,
\end{align}
where $v \prec u$ if $\omega(v) < \omega(u)$ or if weights are equal, we use some
fixed linear order on vertices.

For the R-GMWCS we set root of the arborescence to be the same as the instance
root.

\subsubsection{Restricting traversal.}

Moreover, connected graph can be traversed from the same vertex in 
different ways.
In this section we show how to make infeasible such solutions that could not
be reached by a breadth-first search (BFS).

To achieve such form of the arborescence we add constraints:
\begin{align}
    \label{in:bfsforward}
    d_v - d_u &\le n - (n - 1)w_e, &\forall e = (v, u) \in E; \\
    \label{in:bfsback}
    d_u - d_v &\le n - (n - 1)w_e, &\forall e = (v, u) \in E.
\end{align}

These constraints state that if an edge $e$ is present the subgraph 
then the distances to endpoints differ by one.

\begin{proposition}
    For any connected subgraph $G_s$ of the graph $G$ there exists a 
    solution 
    $(\overline{r}, \overline{y}, \overline{w}, \overline{x}, \overline{d})$
    that encodes subgraph $G_s$ and is feasible to 
\eqref{in:edge}-\eqref{in:treeoversg},
    \eqref{in:linrootdist}-\eqref{in:lindist2}
    and 
    \eqref{in:maxroot}-\eqref{in:bfsback}.
\end{proposition}
\begin{proof}
    First, for any subgraph $G_s$ we can select any of its vertices,
    in particular one with the maximal weight,
    and make a BFS traversal starting from that vertex.
    As was shown above for any connected subgraph $G_s$ and any
    its arborescence there is a corresponding encoding 
    $(\overline{r}, \overline{y}, \overline{w}, \overline{x}, \overline{d})$
    that satisfy constraints 
    \eqref{in:edge}-\eqref{in:treeoversg} and 
    \eqref{in:linrootdist}-\eqref{in:lindist2}.
    By selection of the vertex with the maximal weight as an arborescence
    root constraint \eqref{in:maxroot} holds. Constraints
    \eqref{in:bfsforward}-\eqref{in:bfsback}
    also hold as they directly follow from the BFS ordering.
    \hfill$\square$
\end{proof}

\subsection{Extra cuts}

We also use additional cuts, similar to ones proposed by {\'A}lvarez-Miranda et
al. in~\cite{Alvarez2013}. However we modified them for applying in edge-based
R-GMWCS problem. Such cuts are useful for decreasing the upper bound of the 
objective of the MIP problem that solving using brunch and cut algorithm.

So, we use cut constraints of the form:
\begin{align}
    y_{v} &\le \sum_{e \in C} w_e, &\forall v \in V, C \in C^*,
\end{align}
where $C^*$ is a set of all cuts between instance root and vertex $v$. To
find violated inequalities we associate with each edge $e$ the LP relaxation value
of the variable $w_e$ and for each vertex $v$ we try to find violated
constraint by looking for the cut $C$ such that $y_v > \sum_{ e \in C}
w_e$. The minimum cut is the best candidate of constraint being violated. So, we
use the Edmonds-Karp algorithm to find such violated constraints.

\section{Experimental results}
\label{sec:experiments}

As a testing dataset we used 101 instance
generated by Shiny GAM, a web-service for integrated transcriptional and
metabolic network analysis \cite{Sergushichev2016}. In the dataset,
there are 38 instances of node-weighted SMWCS and 63 instances of GMWCS.
Archive with instances is available at 
\url{http://genome.ifmo.ru/files/papers_files/WABI2016/gmwcs/instances.tar.gz}.

For the comparison we selected two other solvers: \emph{Heinz} version 1.68
\cite{Dittrich2008a} and \emph{Heinz2} version 2.1 \cite{El-Kebir2014}.
The first one, \emph{Heinz}, was initially developed 
for node-weighted SMWCS, but later was adjusted to account for edge weights,
however, only acyclic solutions are considered. The second one, \emph{Heinz2},
does not accept edge weights, but works faster 
than \emph{Heinz} on node-weighted instances.

We ran each of the solver on each of the instances for 10 times with
a time limit of 1000 seconds. \emph{Heinz2} and our GMWCS solver were run
using 4 threads. The processor was AMD Opteron 6380 2.5GHz. 
A table with the results table are available at 
\url{http://genome.ifmo.ru/files/papers_files/WABI2016/gmwcs/results.final.tsv}. 

\subsection{Results for simple MWCS}

The experiments have shown that on the node-weighted instances
GMWCS solver has a performance similar to \emph{Heinz2}
(Fig. \ref{fig:comparisons}A).
For 24 instances (63\%) GMWCS is slower than \emph{Heinz2}.
However, 32 instances (84\%) were solved by GMWCS within
30 seconds, compared to 27 (71\%) of \emph{Heinz2}. Moreover,
4 instances were not solved by \emph{Heinz2} in the allowed time of
1000 seconds compared to only 1 instance for GMWCS.

\begin{figure}[H]
\centering
\includegraphics[width=12cm]{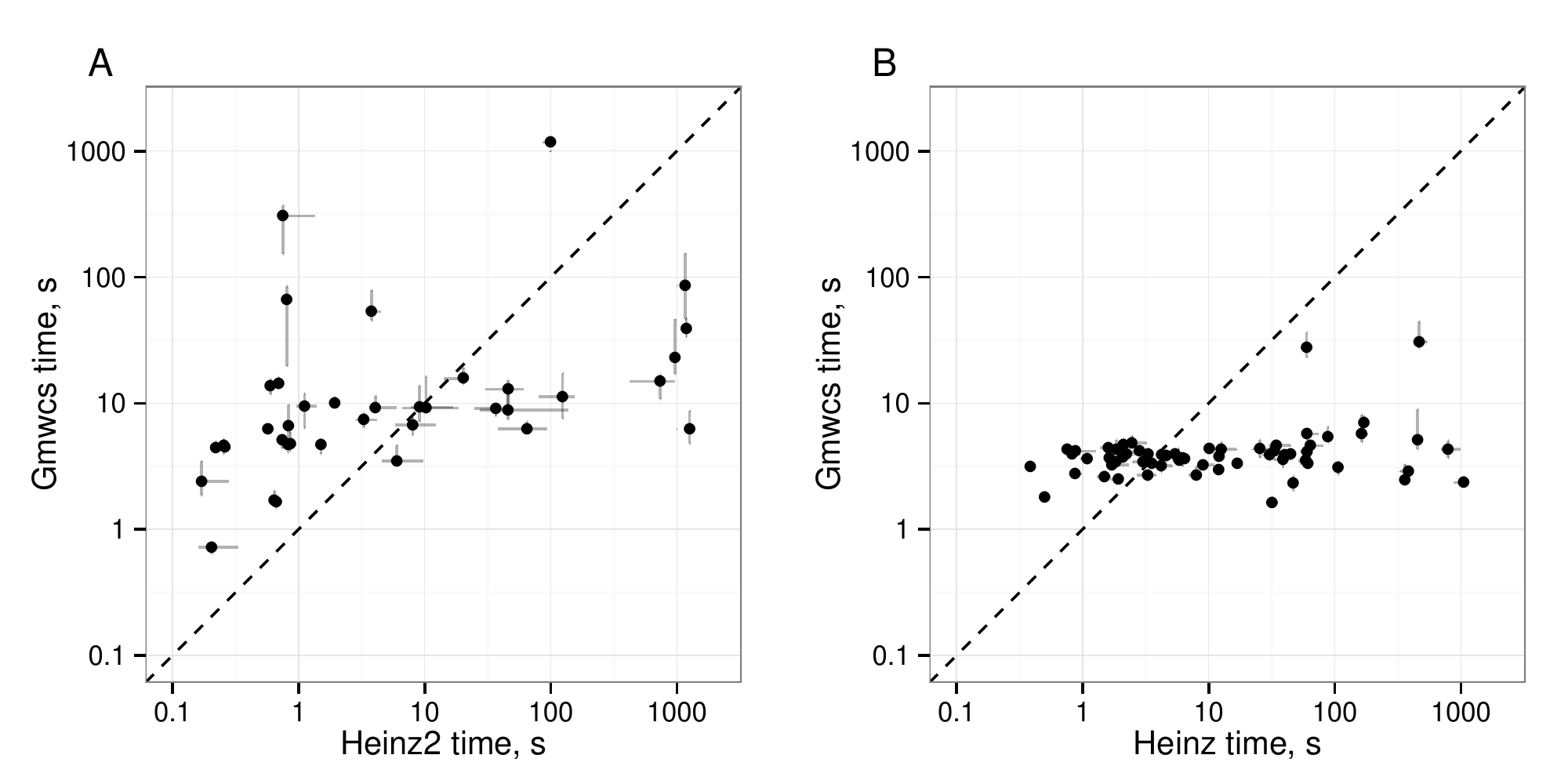}%
\caption{Comparison of GMWCS with \emph{Heinz2} and \emph{Heinz}
    solvers on
    node-weighted (A) and node-and-edge-weighted (B) instances.
    The points represent median times of 10 runs on one instance. 
    Horizontal and vertical grey lines represent the second minimal
    and the second maximal times.
    For convenience a small random noise was added
    to the median values of more then 950 seconds.
    }
\label{fig:comparisons}
\end{figure}

\subsection{Results for generalized MWCS}

For the edge-weighted GMWCS instances
GMWCS solver was able to find optimal solutions within 10 seconds 
all instances except two, while
it took for \emph{Heinz} more than 10 seconds to solve 30 of the instances 
(48\%) (Fig. \ref{fig:comparisons}B).
Moreover, only 35 instances (56\%) had an acyclic solution, accordingly,
28 instances were not solved to GMWCS-optimality by \emph{Heinz}.

\section{Conclusion}

Network analysis approaches are being actively developed for analyzing 
biological data. From the mathematical point of view this usually correspond
to NP-hard problems. Here we described an exact practical solver for 
a particular formulation of generalized maximum weight
connected subgraph problem that naturally arises in metabolic networks.
We have tested the method on the real-world data and have shown
that the developed solver is similar in performance to an existing
solver \emph{Heinz2} on a simple MWCS instances and works better and
more accurately compared to \emph{Heinz} on the edge-weighted instances.
The implementation is freely available at 
\url{https://github.com/ctlab/gmwcs-solver}.

\section*{Funding}

This work was supported by Government of Russian Federation [Grant 074-U01 to A.A.S., A.A.L.].

\bibliography{paper}{}
\bibliographystyle{splncs03}

\end{document}